\documentclass[aps,prl,twocolumn,superscriptaddress]{revtex4-2}
\usepackage{amssymb,amsmath,amsthm}
\usepackage{graphicx}
\usepackage{bm,bbm}
\usepackage[bookmarks=false]{hyperref}
\hypersetup{colorlinks=true,citecolor=blue,linkcolor=red,%
urlcolor=blue,pdfstartview=FitH,bookmarksopen=true}
\usepackage{comment}

\usepackage{algorithm}
\usepackage{algcompatible}

\newcommand{\ket}[1]{\mbox{$ | #1 \rangle $}}
\newcommand{\bra}[1]{\mbox{$ \langle #1 | $}}

\newcommand{\tr}{\mathrm{tr}}
\newcommand{\cM}{\mathcal{M}}
\newcommand{\cF}{\mathcal{F}}
\newcommand{\cC}{\mathcal{C}}

\newtheoremstyle{note}
  {\topsep/2}              	
  {\topsep/2}            	
  {}                        
  {\parindent}             	
  {\itshape}                
  {.---}                    
  {0pt}                     
  {\thmname{#1}\thmnumber{ \itshape#2}\thmnote{ (#3)}} 

\newtheorem{theorem}{Theorem}

\newtheorem{corollary}{Corollary}
\newtheorem{proposition}[theorem]{Proposition}

\theoremstyle{definition}

\theoremstyle{remark}

\begin{document}
\title{Universally Optimal Verification of Entangled States with Nondemolition Measurements}

\author{Ye-Chao Liu}
\affiliation{Key Laboratory of Advanced Optoelectronic Quantum Architecture and Measurement of
Ministry of Education, School of Physics, Beijing Institute of Technology, Beijing 100081, China}

\author{Jiangwei Shang}
\email{jiangwei.shang@bit.edu.cn}
\affiliation{Key Laboratory of Advanced Optoelectronic Quantum Architecture and Measurement of
Ministry of Education, School of Physics, Beijing Institute of Technology, Beijing 100081, China}

\author{Rui Han}
\affiliation{Centre for Quantum Technologies, National University of Singapore, Singapore 117543, Singapore}

\author{Xiangdong Zhang}
\email{zhangxd@bit.edu.cn}
\affiliation{Key Laboratory of Advanced Optoelectronic Quantum Architecture and Measurement of
Ministry of Education, School of Physics, Beijing Institute of Technology, Beijing 100081, China}

\date{\today}
%

\begin{abstract}
The efficient and reliable characterization of quantum states plays a vital role in most, if not all,
quantum information processing tasks.
In this work, we present a universally optimal protocol for verifying entangled states by employing
the so-called quantum nondemolition measurements,
such that the verification efficiency is equivalent to that of the optimal global strategy.
Instead of being probabilistic as the standard verification strategies,
our protocol is constructed sequentially, which is thus more favorable for experimental realizations.
In addition, the target states are preserved in the protocol after each
measurement, so can be reused in any subsequent tasks.
We demonstrate the power of our protocol for the optimal verification of Bell states,
arbitrary two-qubit pure states, and stabilizer states.
We also prove that our protocol is able to perform tasks including fidelity estimation
and state preparation.
\end{abstract}

\maketitle
%

\textit{Introduction.---}%
One basic yet important step in almost all quantum information processing tasks is to efficiently
and reliably characterize the quantum states.
However, the standard tool of quantum state tomography \cite{QSE2004} is typically rather time consuming
and computationally hard due to the exponentially increasing number of parameters to be reconstructed
\cite{Haffner.etal2005,Shang.etal2017}.
Thus, much attention has been drawn to the quest for nontomographic methods
\cite{Mayers.etal2004,Toth.Guehne2005c,Guehne.Toth2009,Flammia.Liu2011,Dimic.Dakic2018},
among which quantum state verification (QSV) \cite{Pallister.etal2018} particularly stands out because of
its many notable properties including its high efficiency and low cost of resources.
Up to now, a large variety of bipartite and multipartite quantum states
\cite{Pallister.etal2018, Hayashi.etal2006, Morimae.etal2017,Takeuchi.Morimae2018,
Yu.etal2019,Li.etal2019,Wang.Hayashi2019,Zhu.Hayashi2019a,
Zhu.Hayashi2019b,Liu.etal2019b,Li.etal2020b,Zhu.Hayashi2019c,
Zhu.Hayashi2019d,Zhang.etal2020a,Jiang.etal2020,Zhang.etal2020b,Li.etal2020a,Dangniam.etal2020}
can be verified efficiently or even optimally by QSV.
Very recently, efficient protocols for verifying quantum processes (including quantum gates and
quantum measurements) have also been proposed \cite{Liu.etal2020a,Zhu.Zhang2019,Zeng.etal2020}.

In short, QSV is a procedure for gaining confidence that the output of a quantum
device is a particular target state $\ket{\psi}$ over any others using local measurements.
A QSV protocol $\Omega$ takes on the general form
\begin{equation}
  \Omega = \sum_i \mu_i\Omega_i\,,
\end{equation}
where $\{\Omega_i,\openone-\Omega_i\}$ is a set of two-outcome tests carried out with probability
$\{\mu_i\}$. The projective operators $\Omega_i$s satisfy $\Omega_i\ket{\psi}=\ket{\psi}$ for all $i$.
If all $N$ states passed the test, we achieve the confidence level $1-\delta$ with
$\delta\le[1-\epsilon \nu(\Omega)]^N$,
where $\epsilon$ is the infidelity of the states and $\nu(\Omega):=1-\lambda_2(\Omega)$ denotes
the spectral gap between the largest and the second largest eigenvalues of $\Omega$
\cite{Pallister.etal2018, Zhu.Hayashi2019c}.
Hence, the protocol $\Omega$ can verify $\ket{\psi}$ to infidelity
$\epsilon$ and confidence level $1-\delta$ with the number of copies of the
states satisfying
\begin{equation}\label{eq:QSVparameter}
  N\geq\frac{\ln\delta^{-1}}{\ln\bigl\{[1-\nu(\Omega)\epsilon]^{-1}\bigr\}}\approx
  \frac1{\nu(\Omega)}\epsilon^{-1}\ln\delta^{-1}\,.
\end{equation}

Compared with tomography as well as other nontomographic methods, properly engineered QSV
protocols can greatly reduce the cost of resources.
Additionally, $\Omega_i$s are expected to be implementable
with local measurements only, thus facilitating the ease of experimental realizations.
However, there remains three major issues associated with QSV.
The first one concerns its efficiency, where an optimal strategy can rarely
(if not impossible at all) be devised.
Second, the measurements $\{\Omega_i\}$ are implemented in a probabilistic manner with
probability distribution $\{\mu_i\}$, which can be very difficult to handle in experiments since the random changes of measurement settings are error prone.
Thus, in practice, instead of choosing $\{\Omega_i\}$ randomly, they are performed in a pre-chosen sequence with the number of measurements for each setting proportional to the ratio of the probability distribution $\{\mu_i\}$ \cite{Zhang.etal2020a,Jiang.etal2020,Zhang.etal2020b}.
Such a compromise would lead to a malicious adversary cheat if one knows the sequence of measurements \cite{Zhu.Hayashi2019c,Zhu.Hayashi2019d}.
Lastly, the unknown quantum states to be characterized are destroyed after each measurement
as the system collapses at the detector, thus, cannot be reused in any subsequent tasks.
In fact, the latter two issues with QSV also exist in tomography and
other nontomographic methods.

In this work, we propose a new type of protocol to tackle \emph{all} the issues associated with QSV.
Our protocol is based on the so-called quantum nondemolition (QND) measurement
\cite{Thorne.etal1978,Braginsky.etal1980,Ralph.etal2006}, which is the type of measurement
that leaves the post-measurement quantum states undestroyed,
thus allowing repeated or sequential measurements.
We fully explore the use of sequentially constructed QND measurements for state verification
instead of the probabilistic construction as in standard QSV.
Under such a scheme, not only can we preserve the target states, but also
manage to universally reach the efficiency of the optimal global strategy.
Specifically, in order to verify the target state within infidelity $\epsilon$ and confidence
level $1-\delta$, we only need $N\approx\epsilon^{-1}\ln\delta^{-1}$ copies of the states.
In addition, our protocol is robust in the sense that the sequence of measurements can be constructed
in an arbitrary order which is rather friendly to experimental implementations.
We demonstrate the power of our protocol for the optimal verification of Bell states, arbitrary two-qubit
pure states, and stabilizer states.
Last but not least, we prove that the protocol can also be used to perform other tasks including
fidelity estimation and state preparation.

\textit{Nondemolition quantum verification.---}%
The QND measurements are often realized through the entanglement with an ancilla system
followed by a measurement on the ancilla.
Let us consider the joint system-ancilla state
$\ket{\psi}\otimes\ket{0}$, where the ancilla qubit is initially prepared in state $\ket{0}$.
Next, we entangle the system and the ancilla via
\begin{equation}\label{eq:U_define}
    U_i=\Omega_i\otimes\openone+(\openone-\Omega_i)\otimes X\,,
\end{equation}
with $X$ being the Pauli-$X$ operator.
Then we obtain
\begin{equation}
    \ket{\psi_i}:=\Omega_i\ket{\psi}\otimes\ket{0}+(\openone-\Omega_i)\ket{\psi}\otimes\ket{1}\,,
\end{equation}
where $\{\Omega_i,\openone-\Omega_i\}$ are the ``pass-or-fail" tests for verifying $\ket{\psi}$ in standard QSV.
Note that the unitarity of $U_i$ is ensured since $\Omega_i$ is a projector.
With this operation, performing a Pauli-$Z$ measurement on the ancilla qubit of the coupled state $\ket{\psi_i}$
is effectively equivalent to the realization of the two-outcome measurement
$\{\Omega_i,\openone-\Omega_i\}$ on the system.
This procedure can be concisely described by the operation
\begin{equation}\label{eq:Mi_define}
    \cM_i=\left(\openone\otimes\ket{0}\bra{0}\right)U_i
\end{equation}
on the joint system $\ket{\psi}\otimes\ket{0}$, which is a QND measurement on $\ket{\psi}$.
Note that $\cM_i$ corresponds to a positive operator-valued measure, and is usually not Hermitian.

It can be easily checked $\cM_i\bigl(\ket{\psi}\otimes\ket{0}\bigr)=\ket{\psi}\otimes\ket{0}$ for all $i$.
Therefore, verifying the target state $\ket{\psi}$ by $\Omega_i$ is exactly the same as
verifying $\ket{\psi}$ nondestructively using $\cM_i$.
In this way, we reformulate the procedure of QSV using QND measurements with the addition
of ancilla qubits, which we dub as \emph{nondemolition quantum verification} (NDQV).
Here we have two remarks. First, due to the dichotomic nature of the measurements
$\{\Omega_i,\openone-\Omega_i\}$ on the system, the coupled ancilla can always be chosen
as a two-dimensional qubit, no matter what the dimension of the target system is \cite{Guo.etal2001}.
Second, the entangling operation $U_i$ is of a similar structure to those used in many other applications
like quantum error correction \cite{Gottesman1997, Preskill1998, Chiaverini.etal2004,
Knill2005, QEC2013} and can be realized with standard quantum gates.

\textit{Sequential NDQV.---}%
The NDQV protocol can, of course, be implemented in a probabilistic manner as that of the standard QSV.
However, doing so would be a waste of resources as the key
advantage of the NDQV protocol lies in the fact that
the target state is not destroyed and remains undisturbed as long as the test passes.
Thus, the post-measurement state can be reliably used and measured again.
Forasmuch, we introduce the \emph{sequential} NDQV protocol.
\begin{theorem}\label{thm:SNDQV}
If a target state $\ket{\psi}$ can be verified by the protocol $\Omega=\sum_i\mu_i\Omega_i$,
where $\Omega_i$s are local projectors,
then it can be verified optimally by
\begin{equation}\label{eq:cM}
  \cM=\prod_{i}\cM_i
\end{equation}
with $\cM_i$ being defined by Eqs.~\eqref{eq:Mi_define} and \eqref{eq:U_define}.
The spectral gap of $\cM$ is given by
\begin{equation}
  \nu(\cM)=1\,,
\end{equation}
indicating that the verification efficiency of $\cM$ is the same as that of the optimal global strategy.
\end{theorem}
\begin{proof}
Here, we briefly sketch the proof and the full proof is presented in Appendix~A of the Supplemental Material \cite{supp}.
The protocol $\cM$ with $l$ sequential measurement settings can be written into a summation form, i.e.,
\begin{equation}\label{eq:M_define}
    \cM\!=\!
    \sum_{k_1k_2\cdots k_l}
  \!\Big\{\!\prod_i \!\Big[\frac{\openone}{2}+(-1)^{k_i}\frac{2\Omega_i-\openone}{2}\Big]\otimes \mathcal{Q}_{k_1k_2\cdots k_l}\!
  \Big\},
\end{equation}
where $\mathcal{Q}_{k_1k_2 \cdots k_l}:=\ket{00\cdots0}\bra{k_1 k_2 \cdots k_l}$ with $k_i\in\{0,1\}$
is an operator in the ancilla subspace.
With the ancilla qubits initially prepared in state $\ket{0}^{\otimes l}$, one gets
the spectral gap $\nu(\cM):=1-\lambda_2(\cM)$, where $\lambda_2(\cM)$ is the second largest eigenvalue of $\cM$.
The corresponding eigenvector of $\lambda_2(\cM)$ is given by $\ket{\psi^\perp}\otimes\ket{0}^{\otimes l}$ which satisfies $\langle\psi^\perp|\psi\rangle=0$.
Then, direct calculation can verify the relation $\nu(\cM)=\nu(\Omega_s)$, where
\begin{equation}\label{eq:Omega_s}
  \Omega_s:=\prod_i\Omega_i=\ket{\psi}\bra{\psi}\,,
\end{equation}
which is demanded in standard QSV.
As a result, $\nu(\cM)=1$.
\end{proof}

Several remarks are in order.
First, following Eq.~\eqref{eq:M_define}, we have
\begin{eqnarray}\label{eq:equiv}
  \cM\Bigl[\sigma\otimes\bigl(\ket{0}\bra{0}\bigr)^{\otimes l}\Bigr]=\Omega_s\sigma \otimes \bigl(\ket{0}\bra{0}\bigr)^{\otimes l}
\end{eqnarray}
for verifying an arbitrary state $\sigma$.
Thus, we can write the sequential NDQV measurement as
\begin{equation}\label{eq:global}
   \cM\widehat{=}\Omega_{s}\otimes\openone\,,
\end{equation}
where the symbol ``$\widehat{=}$" denotes the conditional equivalence when the ancilla register is initially prepared
in state $\ket{0}^{\otimes l}$, which is always the case for the sequential NDQV protocol.
In addition, the order of the measurements $\cM_i$s can be made arbitrary.
This property is a direct consequence of Eq.~\eqref{eq:Omega_s}, the form of which is independent of the order of $\Omega_i$s.
This makes the sequential NDQV protocol rather friendly to experimental implementations.

Second, the optimality of the sequential NDQV protocol directly leads to
the following corollary regarding efficiency.
\begin{corollary}\label{cor:NoMoreMeasure}
The verification efficiency of the sequential NDQV protocol will not be improved by adding
more measurement settings.
\end{corollary}
Rather than taking it as a straightforward consequence of the optimality property, we prove this corollary by
direct calculations in Appendix~B \cite{supp}.
This property of the sequential NDQV protocol is very different from that of the standard QSV strategies
where more measurement settings usually can improve the verification efficiency \cite{Pallister.etal2018}.
In other words, our protocol provides a universal upper bound for the minimal number of
measurement settings demanded for state verification.

Next, the sequential NDQV protocol offers two additional by-products, namely, fidelity estimation
and state preparation.
\begin{corollary}\label{cor:fidEst}
The average fidelity between the unknown state $\sigma$ and the target state $\ket{\psi}$,
i.e., $\cF=\langle F\rangle$, can be directly estimated by the sequential NDQV protocol,
\begin{equation}\label{eq:fidelity}
    \cM\Bigl[\sigma\otimes\bigl(\ket{0}\bra{0}\bigr)^{\otimes l}\Bigr]\cM^\dagger
    =F\ket{\psi}\bra{\psi}\otimes\bigl(\ket{0}\bra{0}\bigr)^{\otimes l}\,,
\end{equation}
where $F=\bra{\psi}\sigma\ket{\psi}$ is the fidelity between $\sigma$ and $\ket{\psi}$.
\end{corollary}
This corollary can be proved by direct calculations using Eqs.~\eqref{eq:Omega_s} and \eqref{eq:global}.
One notices that the resulting output state after the measurement $\mathcal{M}$ has a successful rate $F=\bra{\psi}\sigma\ket{\psi}$, which is also the fidelity between the unknown state $\sigma$ and the target state $\ket{\psi}$.
This indicates that, as long as the measurement $\mathcal{M}$ is successful, the system must be in the target state $\ket{\psi}$. 
The probability of success is given by the fidelity $F$.
Thus, we can estimate $\cF\!=\!\langle F\rangle$ according to the statistical average of the successful rate of the verification.

In addition, the sequential NDQV protocol can be regarded as a state preparation process with successful rate $\cF\!=\!\langle F\rangle$; see Appendix~C \cite {supp} for more details.
In order to prepare one single target state with a possibly malicious provider, other approaches, including the adversarial scenario discussed in standard QSV \cite{Zhu.Hayashi2019c,Zhu.Hayashi2019d} and the related work within the realm of measurement-based quantum computing \cite{Hayashi.Morimae2015,Takeuchi.etal2019}, require additionally a polynomial number of copies of the states.
Our protocol, instead, is able to preserve the target state undisturbed with a higher efficiency without costing additional state preparations.
It enables the ``real-time verification'', in the sense that the state output from the source, being the target state or not, will be projected to the target state with a certain probability, which in turn can be immediately used for any subsequent applications without requiring any further operations or additional state preparations.
This potentially very useful feature is way beyond the reach of other methods.

Lastly, we emphasize that in sequential NDQV protocol, the repeated preparation of the unknown state
to be verified is replaced by the repeated preparation of an ancilla qubit in state $\ket{0}$.
This significantly simplifies the procedure in most experimental scenarios as the dimension of the ancilla qubit
is small and can be initialized much more efficiently.
In the case that repeated use of an ancilla is allowed, either by nondemolishing ancilla measurement
or fast ancilla repreparation, only one single ancilla qubit is physically needed to implement the sequential NDQV protocol.
Moreover, local measurements in different basis on the system, which can be difficult to implement,
are also replaced by the simple Pauli-$Z$ measurement on the ancilla.
In addition, the verified state $\ket{\psi}$ is preserved by the sequential NDQV protocol for future tasks.
All of these come with the cost of implementing the gate operation $U_i$s between the system and the ancilla.
With the booming effort that many research groups and industrial companies are putting into the development of
implementing high fidelity quantum gate operations in various physical systems, we believe that the implementation of the sequential
NDQV protocol will surely become more and more efficient.
For a more thorough discussion on the resource overhead, see Appendix C in the Supplemental Material \cite{supp}.

\textit{Bell state verification.---}%
Consider the case of verifying the Bell state $\ket{\Phi}=(\ket{00}+\ket{11})/\sqrt{2}$.
In the standard QSV protocol $\Omega_{\mathrm{Bell}}$, this state can be verified efficiently
using two measurement settings \cite{Pallister.etal2018}
\begin{equation}
\begin{aligned}
  \Omega_1&=P_{ZZ}^{+}=\ket{0}\bra{0}\otimes\ket{0}\bra{0}+\ket{1}\bra{1}\otimes\ket{1}\bra{1}\,,\\
  \Omega_2&=P_{XX}^{+}=\ket{+}\bra{+}\otimes\ket{+}\bra{+}+\ket{-}\bra{-}\otimes\ket{-}\bra{-}\,,
\end{aligned}
\end{equation}
where $\ket{\pm}=(\ket{0}\pm\ket{1})/\sqrt{2}$.
Taking $\Omega_{\mathrm{Bell}}=\frac1{2}(\Omega_1+\Omega_2)$,
the corresponding spectral gap is  $\nu(\Omega_{\mathrm{Bell}})=\frac{1}{2}$.

\begin{figure}[t]
    \includegraphics[width=1.0\columnwidth]{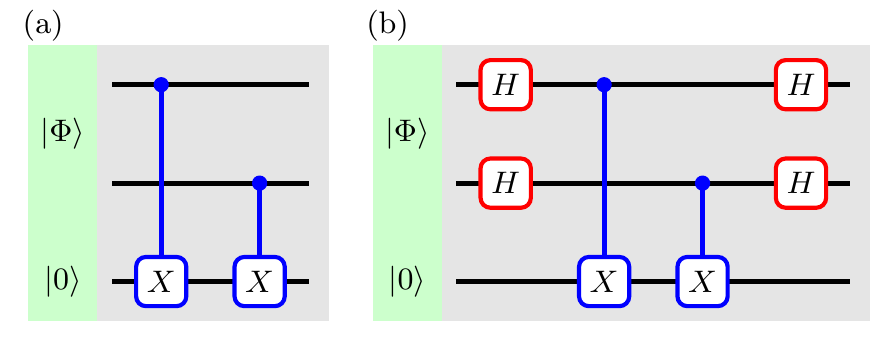}
    \caption{Circuits for the experimental realization of the coupling operations $U_1$ (a) and $U_2$ (b)
    as in Eq.~\eqref{eq:QSV_Bell}.
    $\ket{\Phi}$ is the target Bell state to be verified, and $\ket{0}$ represents the ancilla qubit.}
    \label{fig:circuits}
\end{figure}

Using Theorem~\ref{thm:SNDQV}, we construct the sequential NDQV protocol
for verifying $\ket{\Phi}$ as $\cM_{\mathrm{Bell}}=\cM_1\cM_2$ (or equivalently as
$\cM_{\mathrm{Bell}}=\cM_2\cM_1$).
The two QND measurement settings $\cM_{1(2)}$ are defined as in Eq.~\eqref{eq:Mi_define} with
\begin{equation}
\begin{aligned}
  U_1&=P_{ZZ}^{+}\otimes\openone+\left(\openone-P_{ZZ}^{+}\right)\otimes X\\
  &=\cC_{X1a}\cC_{X2a}\,,\\
  U_2&=P_{XX}^{+}\otimes\openone+\left(\openone-P_{XX}^{+}\right)\otimes X\\
  &=\bigl(H\otimes H\otimes \openone\bigr)\cC_{X1a}\cC_{X2a}\bigl(H\otimes H\otimes \openone\bigr),
  \label{eq:QSV_Bell}
\end{aligned}
\end{equation}
where $H$ is the Hadamard gate, and $\cC_{Xja}$ denotes the \textsc{cnot} gate that the ancilla qubit $a$
is controlled by the $j$th qubit.
The corresponding circuits for the experimental realization of the couplings $U_{1(2)}$
are illustrated in Fig.~\ref{fig:circuits}.

Specifically, the sequential NDQV protocol for verifying $\ket{\Phi}$ proceeds as follows.
Together with an ancilla qubit prepared in state $\ket{0}$, the actual state $\sigma$ is sent into the circuit
for the coupling operation $U_1$,
followed by a Pauli-$Z$ measurement on the ancilla.
If the measurement outcome is $\ket{0}$, together with another freshly prepared ancilla qubit
in state $\ket{0}$, the system state is passed on to the circuit for the coupling operation $U_2$
followed by a Pauli-$Z$ measurement on the ancilla.
If the outcome is still $\ket{0}$, we declare the test passes.
In any other cases, we say that the test fails.
In this way, the protocol $\cM_{\mathrm{Bell}}$ has a spectral gap $\nu(\cM_\mathrm{Bell})=1$,
which is equivalent to the optimal global strategy.

One final remark concerning the verification efficiency is that, in standard QSV,
the efficiency can be further improved to $\nu(\Omega_{\mathrm{Bell}}')=\frac{2}{3}$
by adding an additional measurement setting $P_{YY}^{-}$ \cite{Pallister.etal2018,Liu.etal2019b}.
However, by Corollary~\ref{cor:NoMoreMeasure}, more measurement settings will not help improve
the sequential NDQV protocol as it is already equivalent to the
optimal global strategy.

\textit{Verification of arbitrary two-qubit pure states.---}%
Without loss of generality, we write the two-qubit entangled pure state as
$\ket{\Psi}=\sin\theta\ket{00}+\cos\theta\ket{11}$ with $\theta\in(0,\pi/4)$.
In standard QSV, this state can be verified efficiently using three settings~\cite{Liu.etal2019b},
\begin{equation}
\begin{aligned}
  \Omega_1&=P_{ZZ}^{+}=\ket{0}\bra{0}\otimes\ket{0}\bra{0}+\ket{1}\bra{1}\otimes\ket{1}\bra{1}\,,\\
  \Omega_2&=\openone-\ket{+}\bra{+}\otimes\ket{\varphi_{+}}\bra{\varphi_{+}}\,,\\
  \Omega_3&=\openone-\ket{-}\bra{-}\otimes\ket{\varphi_{-}}\bra{\varphi_{-}}\,,
  \label{eq:QSV_2qb}
\end{aligned}
\end{equation}
where $\ket{\varphi_{\pm}}=\cos\theta\ket{0}\mp\sin\theta\ket{1}$.
The corresponding spectral gap is $\nu(\Omega_{\mathrm{2qb}})=\frac1{3}$ by taking
$\Omega_{\mathrm{2qb}}=\frac1{3}\sum_{i=1}^3\Omega_i$.
The efficiency can be improved by modified QSV protocols with different measurement settings.
More details can be found in Appendix~D \cite{supp},
where we also discuss how to perform verification using adaptive methods.
Then in Appendix~E \cite{supp}, we show how to realize the corresponding adaptive QND measurements.

By applying Eqs.~\eqref{eq:Mi_define} and \eqref{eq:U_define}, we find that the QND implementation
of Eq.~\eqref{eq:QSV_2qb} is given by
\begin{equation}
\begin{aligned}
 \cM_{1}
  &=\bigl(\openone\otimes\ket{0}\bra{0}\bigr)\cC_{X1a}\cC_{X2a},\\
  \cM_{i}^{t}
  &=\bigl(\openone\otimes\ket{0}\bra{0}\bigr)\bigl(R_i^{\dagger}\otimes\openone\bigr)\bigl(X\otimes X\otimes\openone\bigr)\\
  &\,\,\,\quad\cC^2_{X12a}\bigl(X\otimes X\otimes\openone\bigr)\bigl(R_i\otimes \openone\bigr)\,,
\end{aligned}
\end{equation}
for $i=\{2,3\}$, where the rotation matrix $R_{2(3)}$ turns the state $\ket{+}\otimes\ket{\varphi_{+}}$
($\ket{-}\otimes\ket{\varphi_{-}}$) into $\ket{00}$. Specifically,
\begin{equation}
  R_2=H\otimes \left[\begin{smallmatrix}\cos\theta&-\sin\theta\\\sin\theta&\cos\theta\end{smallmatrix}\right]\;
\mathrm{and}\;
  R_3=XH\otimes \left[\begin{smallmatrix}\cos\theta&\sin\theta\\-\sin\theta&\cos\theta\end{smallmatrix}\right].
\end{equation}
These QND measurements require a Toffoli (\textsc{ccnot}) gate $\cC^2_{X12a}$ which is a three-body coupling operation.
For systems where the Toffoli gate is not easily accessible, 
one can effectively replace it with two \textsc{cnot} gates using two ancilla qubits initially prepared in $\ket{00}$.
We denote this set of QND measurements by $\cM_{i}^{b}$ and, for $i=\{2,3\}$,
\begin{equation}\label{eq:NDQV_2qb_2}
\cM_{i}^{b}
 =\openone-\bigl(\openone\otimes\ket{00}\bra{00}\bigr)\bigl(R_i^{\dagger}\otimes\openone\bigr)\cC_{X1a}\cC_{X2a'}\bigl(R_i\otimes \openone\bigr)\,.
\end{equation}

The equivalence of the two sequential NDQV protocols $\cM^{t}=\cM_1\cM_2^t\cM_3^t$ and
$\cM^{b}=\cM_1\cM_2^b\cM_3^b$ (or arbitrary permutations of each of the three measurement settings)
can be obtained from $\cM_{i}^t\widehat{=}\cM_{i}^b$ by direct calculations, and both of them are
equivalent to the optimal global strategy.
The latter replaces the Toffoli gate by two \textsc{cnot} gates at the cost of one additional ancilla qubit.
In fact, this equivalence also holds for many-body coupling operations;
see below the proposition with the proof shown in Appendix~F in the Supplemental Material \cite{supp}.

\begin{proposition}\label{prop:2bodyCouple}
  For the specific setting of NDQV where the ancilla is always prepared in $\ket{0}$ and measured
  in the Pauli-$Z$ basis, a generalized $(n+1)$-body Toffoli gate can always be replaced by $n$ two-body
  \textsc{cnot} gates with ancilla qubits initially prepared in $\ket{0}^{\otimes n}$.
\end{proposition}

\textit{Verification of stabilizer states.---}%
Stabilizer states, such as the GHZ states \cite{Gottesman1996,Gottesman1997}, are an important class
of multipartite states.
An $n$-qubit stabilizer state $\ket{\psi}$ can be determined by a stabilizer group $\mathcal{S}$,
where $\mathcal{S}$ is generated by a set of $n$ commuting Pauli operators $\{S_1,\dots,S_n\}$. 
With $S_i\ket{\psi}=\ket{\psi}$
for all $i$, the stabilizer group uniquely defines the state \ket{\psi}.

In standard QSV, an $n$-qubit stabilizer state can be verified with efficiency $\nu=1/n$
using minimally $n$ measurement settings constructed with the stabilizer generators \cite{Pallister.etal2018}.
This verification efficiency can be improved to $\nu=2^{n-1}/(2^n-1)$ if more settings
(like the $2^n-1$ linearly independent stabilizers) are used.
Since the stabilizer generators are Pauli operators, by Theorem~\ref{thm:SNDQV},
the sequential NDQV protocol can be realized using QND measurements with only
two-body couplings $\cC_X$ and suitable local operations.
These QND measurements can be implemented in the same way as the syndrome measurements for
stabilizer quantum error correction codes \cite{QEC2013}.
For a specific example, see Appendix~G \cite{supp}.

\textit{Discussion.---}%
As compared to standard QSV, the sequential NDQV protocol offers two major advantages,
namely, its optimal global efficiency and its robustness in the measurement sequence.
These advantages come with the unavoidable cost of adding additional ancilla qubits
as well as implementing coupling operations between the system and the ancilla.
In a way, we replace the resources from preparing the state to be verified over and over again
by the cheaper preparation of the ancilla state $\ket{0}$ for each QND measurement.
Although not implemented completely locally, the QND measurement is considered to be a standard technique, and has
been demonstrated in various platforms experimentally \cite{Grangier.etal1998, Nogues.etal1999,
Lupascu.etal2007, Geremia.etal2004, Neumann.etal2010, Robledo.etal2011, Nakajima.etal2019}.
Moreover, the requirement of only two-body couplings, as shown by Proposition~\ref{prop:2bodyCouple}, and the sequential nature of our protocol, can greatly simplify the experimental implementations.

One might find our sequential NDQV protocol similar to the adaptive QSV scheme which can also
be regarded as being realized sequentially.
However, they differ clearly in two major aspects. First, the system is directly measured in the standard
adaptive QSV without the implementation of QND measurements.
Second, in the adaptive approach, the choice of the latter measurement depends on the previous
measurement outcomes, whereas the order of measurement settings in the sequential NDQV
protocol can be arbitrary.
Nevertheless, the adaptive scheme can also be realized using QND measurements as demonstrated
in Appendix~E in the Supplemental Material \cite{supp}.

\textit{Conclusion.---}%
We have presented a universally optimal protocol for quantum state verification using QND measurements.
By virtue of the nondestructive feature of the QND measurements, we proposed the sequential NDQV protocol.
Under such a design, not only can we preserve the target states, but also make our protocol
equivalent to the optimal global strategy in terms of the verification efficiency.
Moreover, our protocol is robust in the sense that the order of the sequential measurements can be
arbitrarily constructed which is rather friendly to experimental implementations.
We demonstrated the power of our protocol through three concrete examples.
In addition, we proved that the protocol can also be used to perform tasks including fidelity estimation
and state preparation.
There are many other interesting aspects of the sequential NDQV protocol, which cannot all be covered in this Letter, to be investigated in the future.
In particular, the effect of noisy implementation, weak measurement on the ancilla and the use of entangled ancilla
when the state is shared among different sites are possible interesting directions.
By employing the state-process duality, our protocol can also be extended to verify quantum processes
including quantum gates and measurements.

\acknowledgments
We are grateful to Xiao-Dong Yu and Huangjun Zhu for helpful discussions.
This work was supported by the National Key R\&D Program of China
under Grant No.~2017YFA0303800 and the National Natural Science Foundation of
China through Grants No.~11574031, No.~61421001, and No.~11805010.  J.S. also
acknowledges support by the Beijing Institute of Technology Research Fund
Program for Young Scholars.
The Centre for Quantum Technologies is a Research Centre of Excellence funded by
the Ministry of Education and the National Research Foundation of Singapore.


%

%
\onecolumngrid
\appendix

\section{Appendix A: Proof of Theorem~1}
The sequential NDQV protocol $\cM$ operates on the joint space of the target state space $T$ and the ancilla space $A$.
Following Eqs.~\eqref{eq:Mi_define} and \eqref{eq:U_define}, $\cM$ with $l$ sequential measurement settings
can be rewritten into a summation form, such that
\begin{eqnarray}
  \cM&=&\prod_i \cM_i\nonumber\\
  &=& \prod_i \Big\{
  \big(\openone^{(T)} \otimes \ket{0}\bra{0}^{(i)} \otimes \openone^{(A\backslash i)}\big)
  \big(U_i^{(T \cup i)} \otimes \openone^{(A\backslash i)}\big)
  \Big\} \nonumber\\
  &=& \prod_i \Big\{
  \big(\openone^{(T)} \otimes \ket{0}\bra{0}^{(i)} \otimes \openone^{(A\backslash i)}\big)
  \big[\Omega_i^{(T)} \otimes \openone^{(i)} \otimes \openone^{(A\backslash i)} + (\openone-\Omega_i)^{(T)} \otimes X^{(i)} \otimes \openone^{(A\backslash i)}\big]
  \Big\} \nonumber\\
  &=& \prod_i \Big\{
  \big[\Omega_i^{(T)} \otimes \ket{0}\bra{0}^{(i)} \otimes \openone^{(A\backslash i)}\big]
  +\big[(\openone-\Omega_i)^{(T)} \otimes \ket{0}\bra{1}^{(i)} \otimes \openone^{(A\backslash i)}\big]
  \Big\} \nonumber\\
  &=& \sum_{k_1 k_2\cdots k_l} \bigg\{
  \Big[\prod_i \Big(\frac{\openone}{2}+(-1)^{k_i}\frac{2\Omega_i-\openone}{2}\Big)
  \Big]^{(T)}\otimes \mathcal{Q}_{k_1 k_2\cdots k_l}^{(A)}
  \bigg\}\,.
\end{eqnarray}
where $\mathcal{Q}_{k_1 k_2\cdots k_l}:=\ket{0}\bra{k_1}\otimes\ket{0}\bra{k_2}\otimes\cdots\otimes\ket{0}\bra{k_l}=\ket{00\cdots0}\bra{k_1 k_2 \cdots k_l}$ with  $k_i\in\{0,1\}$.
The subscripts label different measurement settings, whereas the superscripts represent the subspace where the corresponding operator lies on.
The relative complement notation $(A\backslash i)$ means that it operates on the ancilla subspace $A$ except qubit $i$.
Note that for simplicity and without causing any ambiguity, the superscripts are omitted in the main text.

The initialized ancilla register can be described as $\ket{0}^{\otimes l}$, meaning that all $l$ qubits are initialized in state $\ket{0}$.
Then, without loss of generality, we can write the actual output state as $\sigma=\ket{\psi_\epsilon}\bra{\psi_\epsilon}$ with $\ket{\psi_\epsilon}=\sqrt{1-\epsilon}\ket{\psi}+\sqrt{\epsilon}\ket{\psi^\perp}$, where $\ket{\psi^\perp}$
represents an arbitrary state orthogonal to $\ket{\psi}$ ($\langle\psi|\psi^\perp\rangle=0$) and $0<\epsilon<1$
is the infidelity.
Therefore, following the properties of each measurement setting that
\begin{eqnarray}
  \cM_i\bigl(\ket{\psi}\otimes\ket{0}\bigr)
  \ \ &=&\ket{\psi}\otimes\ket{0}\,,
\end{eqnarray}
for all $i$, we get
\begin{eqnarray}
  \cM\bigl(\ket{\psi}\otimes\ket{0}^{\otimes l}\bigr)\ \ &=&\ket{\psi}\otimes\ket{0}^{\otimes l}\,,
\end{eqnarray}
for the entire sequential measurement of $\cM$.
With this, the maximal probability for state $\sigma$, which satisfies $\langle\psi|\sigma|\psi\rangle\leq1-\epsilon$,
to pass the protocol $\cM$ is given by
\begin{eqnarray}
  &&\max_{\langle\psi|\sigma|\psi\rangle\leq1-\epsilon}
  \tr\Bigl[\cM\Bigl(\sigma\otimes\bigl(\ket{0}\bra{0}\bigr)^{\otimes l}\Bigr)\Bigr]\nonumber\\
  &=&\max_{|\psi^\perp\rangle}
  \tr\Big\{
  (1-\epsilon)\cM\Bigl(\ket{\psi}\bra{\psi}\otimes \bigl(\ket{0}\bra{0}\bigr)^{\otimes l}\Bigr)
  +\epsilon\cM\Bigl(\ket{\psi^\perp}\bra{\psi^\perp}\otimes\bigl(\ket{0}\bra{0}\bigr)^{\otimes l}\Bigr)\nonumber\\
  &&\qquad\quad+\sqrt{\epsilon}\sqrt{1-\epsilon}\cM\Bigl[\bigl(\ket{\psi^\perp}\bra{\psi}
  +\ket{\psi}\bra{\psi^\perp}\bigr)\otimes\bigl(\ket{0}\bra{0}\bigr)^{\otimes l}\Bigr]
  \Big\}\nonumber\\
  &=&1-[1-\lambda_2(\cM)]\epsilon\,,
\end{eqnarray}
where $\lambda_2(\cM)$ is the second largest eigenvalue of $\cM$ with the corresponding eigenstate
$\ket{\psi^\perp}\otimes\ket{0}^{\otimes l}$.
Hence, in order to verify the target state $\ket{\psi}$ within infidelity $\epsilon$ and confidence
level $1-\delta$, we need
\begin{equation}
  N\approx\frac{1}{\nu(\cM)}\epsilon^{-1}\ln\delta^{-1}
\end{equation}
copies of the states, where
\begin{equation}
  \nu(\cM):=1-\lambda_2(\cM)
\end{equation}
denotes the spectral gap between the largest and the second largest eigenvalues of $\cM$.

Similar to the standard QSV protocol \cite{Pallister.etal2018}, we also have
\begin{eqnarray}
  &&\cM\Bigl(\sigma\otimes\bigl(\ket{0}\bra{0}\bigr)^{\otimes l}\Bigr)\nonumber\\
  &=&\prod_i\Omega_i\sigma \otimes \mathcal{Q}_{00\cdots0} \bigl(\ket{0}\bra{0}\bigr)^{\otimes l}
  + \sum_{k_1 k_2 \cdots k_l\neq 00\cdots0} \bigg\{
  \Big[\prod_i \big(\frac{\openone}{2}+(-1)^{k_i}\frac{2\Omega_i-\openone}{2}\big)
  \Big] \sigma \otimes \mathcal{Q}_{k_1 k_2 \cdots k_l} \bigl(\ket{0}\bra{0}\bigr)^{\otimes l}
  \bigg\} \nonumber\\
  &=&\prod_i\Omega_i\sigma \otimes \bigl(\ket{0}\bra{0}\bigr)^{\otimes l}\,.
\end{eqnarray}
By defining $\Omega_s:=\prod_i\Omega_i$, we have
\begin{eqnarray}
  \max_{\langle\psi|\sigma|\psi\rangle\leq1-\epsilon}
  \tr\Bigl[\cM\Bigl(\sigma\otimes\bigl(\ket{0}\bra{0}\bigr)^{\otimes l}\Bigr)\Bigr]
  =\max_{\langle\psi|\sigma|\psi\rangle\leq1-\epsilon}
  \tr\big(\Omega_s\sigma\big)
  =1-\Big[1-\lambda_2(\Omega_s)\Big]\epsilon\,,
\end{eqnarray}
so that
\begin{equation}
    \nu(\cM)=\nu(\Omega_s)\,.
\end{equation}

As $\Omega_i$s are local projectors satisfying $\Omega_i\ket{0}=\ket{0}$, they should take on the general form
\begin{equation}
   \Omega_i=\ket{\psi}\bra{\psi}+\sum_{p}\lambda_i^p\ket{\psi_i^p}\bra{\psi_i^p}\,,
\end{equation}
where each set of the basis \{\ket{\psi_i^p}\} spans the subspace orthogonal to $\ket{\psi}$, and $\lambda_i^p=0\text{ or }1,~\forall(\Omega_i,p)$.
Then, one can get
\begin{equation}\label{eq:Omega_s}
  \Omega_s:=\prod_i\Omega_i=\ket{\psi}\bra{\psi}+\prod_i\Bigl(\sum_p \lambda_i^p\ket{\psi_i^p}\bra{\psi_i^p} \Bigr).
\end{equation}
Because no other states except the target one can pass all the measurement settings, that is, if $\lambda_j^q=1$, there exists at least one measurement setting with $\Omega_i\ket{\psi_j^q}=0$, leading to
\begin{eqnarray}
  \sum_p \lambda_i^p\ket{\psi_i^p}\bra{\psi_i^p}\psi_j^q\rangle=\bra{\psi}\psi_j^q\rangle\ket{\psi}=0\,.
\end{eqnarray}
Then we have
\begin{eqnarray}
  \Omega_s=\ket{\psi}\bra{\psi}\,,
\end{eqnarray}
which is independent of the order of the measurement settings $\Omega_i$.
Finally, we obtain
\begin{equation}
    \nu(\cM)=\nu(\Omega_s)=1\,.
\end{equation}

\section{Appendix B: Proof of Corollary~\ref{cor:NoMoreMeasure}}\label{app:NoMoreMeasure}
\begin{proof}
Consider a new measurement setting with the general form $\Omega_j=\lambda_0\ket{\psi}\bra{\psi}+\sum_p\lambda_j^p\ket{\psi^{p}}\bra{\psi^{p}}$, where $0<\lambda_0\leq 1$
as we allow $\Omega_j$ to be a general positive operator-valued measure.
Using Theorem~\ref{thm:SNDQV}, we construct the new sequential protocol as
\begin{equation}
   \cM'=\cM\cM_j\widehat{=}\lambda_0\ket{\psi}\bra{\psi}\otimes\openone\,.
\end{equation}
Hence, the new spectral gap is given by $\nu(\cM')=\lambda_0 \le 1$, meaning that the verification efficiency
is not improved.
\end{proof}

\section{Appendix C: Resource overhead of the sequential NDQV protocol}
Compared to the traditional method of quantum state tomography, QSV is able to improve the scaling of
the characterization efficiency from exponential to polynomial.
On top of this, the sequential NDQV protocol further pushes the efficiency to be equivalent to the optimal global strategy.

Generally speaking, there are two main resource cost in the sequential NDQV protocol.
One is that a large number of ancilla qubits are needed because each measurement in the sequence     requires an ancilla qubit.
However, since our protocol is constructed sequentially, in the case that repeated use of an ancilla qubit is allowed, either by nondemolishing ancilla measurement or fast ancilla reinitialization,
only one ancilla qubit is sufficient.
The other main resource cost is the implementation of the coupling operation in each measurement setting.
As demonstrated in Appendix~F below, all the couplings can be implemented by at most $O(n)$ local operations and $O(n)$ {\sc cnot} gates.
In addition, making use of the available multi-body couplings in experiments can often reduce these costs.

It is also worth noting that, in some scenario, the system state to be verified is distributed among different physically located verifiers via a quantum channel. 
For such a case, one single ancilla qubit might not be sufficient. 
According to Proposition 2, each verifier then needs to prepare an ancilla qubit locally, entangle it to its own system and carry out the local projective measurement on the ancilla. 
All these operations can be done easily by the verifiers and the measurement results can be sent via classical communication channels. 
In this case, the number of ancilla qubits required scales linearly as the number of verifiers.
Moreover, the local preparation of ancilla qubits and the use of classic communication is much more efficient and experimentally friendly than preparing and distributing the high dimensional system state over and over again.
In some special cases, using entangled ancilla qubits and distribute them to the verifiers using quantum channels might enhance the performance. However, this is beyond the scope of our current work and will be a future direction of investigation.

The successful implementation of the QSV protocols, including our sequential NDQV protocol, requires a continuous outputs of the ``pass" outcome.
In real experiment, however, this might not be satisfied due to various problems.
An alternative approach is to modify the data processing method by recording the frequency $f$ of the ``pass'' instances \cite{Yu.etal2019}.
If $f>1-\epsilon\nu$, the confidence level $1-\delta$ can be derived from the Chernoff bound
\begin{eqnarray}
  \delta\leq e^{-D[f||(1-\epsilon\nu)]N}\,,
\end{eqnarray}
where $D(x||y)=x\log(\frac{x}{y})+(1-x)\log(\frac{1-x}{1-y})$ is the Kullback-Leibler divergence.
Several recent QSV experiments \cite{Zhang.etal2020a,Jiang.etal2020,Zhang.etal2020b} employed exactly this approach to process the data,
and the advantages of QSV are demonstrated clearly.
However, how to account for the influence of various noise in experiments is still an open problem.

One may also regard the sequential NDQV protocol as a state preparation process as shown by Corollary~2 in the main text.
However, compared with the methods that generate arbitrary quantum states which usually requires complex unitary operations or even complex networks constructed by a universal set of gates, our protocol has a low complexity and resource consumption since only two-body {\sc cnot} gates, local rotations and local measurements are needed.
This is why state preparation of our protocol is only $\mathcal{F}$-efficient as shown by Corollary~2, so probably purification rather than direct preparation is more useful in practice.

\section{Appendix D: Standard QSV protocols for verifying arbitrary two-qubit pure states}\label{app:QSV_2qubit}
\subsection{1. The non-adaptive approach}
To verify an arbitrary two-qubit pure state $\ket{\Psi}=\sin\theta\ket{00}+\cos\theta\ket{11}$
with $\theta\in(0,\pi/4)$, the standard QSV protocol with the optimal verification efficiency
using only local and non-adaptive measurements contains four measurement settings \cite{Pallister.etal2018},
i.e.,
\begin{eqnarray}
  \Omega^{(4)}=\alpha(\theta)P_{ZZ}^{+}+\frac{1-\alpha(\theta)}{3}\sum_{k=1}^{3}\bigl[\openone-\ket{\phi_k}\bra{\phi_k}\bigr]\,,\quad\quad \mbox{for}~\alpha(\theta)=\frac{2-\sin(2\theta)}{4+\sin(2\theta)}\,,
\end{eqnarray}
with the efficiency giving by $\nu(\Omega^{(4)})=1/(2+\sin\theta\cos\theta)$.
The first setting $P_{ZZ}^{+}=\ket{00}\bra{00}+\ket{11}\bra{11}$ is the projector onto the positive eigenspace
of the Pauli measurement $ZZ$, and the rest three $\openone-\ket{\phi_k}\bra{\phi_k}$
are the measurements that reject the state \ket{\phi_k} where
\begin{eqnarray}
  \ket{\phi_1}&=&\left(\frac{1}{\sqrt{1+\tan{\theta}}}\ket{0}+\frac{e^{\frac{2\pi i}{3}}}{\sqrt{1+\cot{\theta}}}\ket{1}\right)\otimes\left(\frac{1}{\sqrt{1+\tan{\theta}}}\ket{0}+\frac{e^{\frac{\pi i}{3}}}{\sqrt{1+\cot{\theta}}}\ket{1}\right),\\
  \ket{\phi_2}&=&\left(\frac{1}{\sqrt{1+\tan{\theta}}}\ket{0}+\frac{e^{\frac{4\pi i}{3}}}{\sqrt{1+\cot{\theta}}}\ket{1}\right)\otimes\left(\frac{1}{\sqrt{1+\tan{\theta}}}\ket{0}+\frac{e^{\frac{5\pi i}{3}}}{\sqrt{1+\cot{\theta}}}\ket{1}\right),\\
  \ket{\phi_3}&=&\left(\frac{1}{\sqrt{1+\tan{\theta}}}\ket{0}+\frac{1}{\sqrt{1+\cot{\theta}}}\ket{1}\right)\otimes\left(\frac{1}{\sqrt{1+\tan{\theta}}}\ket{0}-\frac{1}{\sqrt{1+\cot{\theta}}}\ket{1}\right).
\end{eqnarray}

However, as we show in Eq.~\eqref{eq:QSV_2qb} of the main text, to verify \ket{\Psi} using only local
and non-adaptive measurements, a minimal three measurement settings are enough, which are
\begin{eqnarray}
  \Omega_1&=&P_{ZZ}^{+}=\ket{0}\bra{0}\otimes\ket{0}\bra{0}+\ket{1}\bra{1}\otimes\ket{1}\bra{1}\,,\nonumber\\
  \Omega_2&=&\openone-\ket{+}\bra{+}\otimes\ket{\varphi_{+}}\bra{\varphi_{+}}\,,\\
  \Omega_3&=&\openone-\ket{-}\bra{-}\otimes\ket{\varphi_{-}}\bra{\varphi_{-}}\,,\nonumber
\end{eqnarray}
where $\ket{\pm}=(\ket{0}\pm\ket{1})/\sqrt{2}$ and $\ket{\varphi_{\pm}}=\cos\theta\ket{0}\mp\sin\theta\ket{1}$.
The latter two $\Omega_{2(3)}$ are the measurements that reject the input states if getting $\ket{+}$ ($\ket{-}$) on the first subsystem and $\ket{\varphi_{+}}$($\ket{\varphi_{-}}$) on the second one simultaneously.
Thus, such a protocol can be constructed as
\begin{eqnarray}
  \Omega^{(3)}=\Omega_{\text{2qb}}=\frac{1}{3}(\Omega_1+\Omega_2+\Omega_3)\,,
\end{eqnarray}
with the efficiency $\nu(\Omega^{(3)})=1/3$ which is independent of the parameter $\theta$
and only a little worse than that of $\Omega^{(4)}$.

\subsection{2. The adaptive approach}
Furthermore, using the general transformation between the adaptive and non-adaptive schemes as presented
in Ref.~\cite{Liu.etal2019b}, the least number of measurement settings for verifying \ket{\Psi} can be reduced
to only two by replacing $\Omega_2$ and $\Omega_3$ in $\Omega^{(3)}$ with a single adaptive measurement
\begin{eqnarray}
  X_{\Psi}=\ket{+}\bra{+}\otimes\ket{\varphi_{+}^{\perp}}\bra{\varphi_{+}^{\perp}}+\ket{-}\bra{-}\otimes\ket{\varphi_{-}^{\perp}}\bra{\varphi_{-}^{\perp}}\,,\label{eq:QSV_adaptive}
\end{eqnarray}
where $\ket{\varphi_{\pm}^{\perp}}=\sin\theta\ket{0}\pm\cos\theta\ket{1}$.
Then we have the protocol
\begin{eqnarray}
  \Omega_\text{adp}^{(2)}=\frac{1}{2}P_{ZZ}^{+}+\frac{1}{2}X_{\Psi}\,,
\end{eqnarray}
which improves the verification efficiency to $\nu(\Omega_\text{adp}^{(2)})=1/2$.

Last but not least, we note that using adaptive measurements, an optimal efficiency can be
achieved by considering the symmetry between the measurement settings.
Such a protocol with three measurement settings has been proposed and proven in Ref.~\cite{Yu.etal2019},
\begin{equation}
  \Omega_\text{adp}^{(3)}=\frac{\cos^{2}\theta}{1+\cos^{2}\theta}P_{ZZ}^{+}
  +\frac1{2(1+\cos^{2}\theta)}X_{\Psi}
  +\frac1{2(1+\cos^{2}\theta)}Y_{\Psi}\,,
\end{equation}
where
\begin{equation}
\begin{aligned}
  P_{ZZ}^{+}&=\ket{0}\bra{0}\otimes\ket{0}\bra{0}+\ket{1}\bra{1}\otimes\ket{1}\bra{1}\,,\\
  X_{\Psi}&=\ket{\varphi_{0}}\bra{\varphi_{0}}+\ket{\varphi_{2}}\bra{\varphi_{2}}\,,\\
  Y_{\Psi}&=\ket{\varphi_{1}}\bra{\varphi_{1}}+\ket{\varphi_{3}}\bra{\varphi_{3}}\,,
\end{aligned}
\end{equation}
with $\ket{\varphi_{0}}=\frac1{\sqrt{2}}(\ket{0}+\ket{1})\otimes(\sin\theta\ket{0}+\cos\theta\ket{1})$
and $\ket{\varphi_{k}}=g^k\ket{\varphi_{0}}$.
The unitary operator $g$ is defined as $g=\Upsilon\otimes\Upsilon^\dagger$, where $\Upsilon$ is the phase gate,
i.e., $\Upsilon\ket{0}=\ket{0}$ and $\Upsilon\ket{1}=\mathrm{i}\ket{1}$.
Note that $X_{\Psi}$ is the same as that in Eq.~\eqref{eq:QSV_adaptive}.
Then the optimal efficiency with adaptive measurements is given by $\nu(\Omega_\text{adp}^{(3)})=1/(1+\cos^{2}\theta)$.

\section{Appendix E: Adaptive QND measurements}\label{app:adaptive}
We use the measurement setting
\begin{equation}
  X_{\Psi}=\ket{+}\bra{+}\otimes\ket{\varphi_{+}^{\perp}}\bra{\varphi_{+}^{\perp}}+\ket{-}\bra{-}\otimes\ket{\varphi_{-}^{\perp}}\bra{\varphi_{-}^{\perp}}
\end{equation}
in Eq.~\eqref{eq:QSV_adaptive} as an example to demonstrate
how to realize adaptive measurements using the nondemolition approach.
Since the measurements in standard QSV protocols are expected to be local, one can first rotate them
to the measurement basis $\{\ket{0},\ket{1}\}$, which can then be realized by the QND measurements
$\bigl\{\bigl(\openone\otimes\ket{0}\bra{0}\bigr)\cC_{X},\bigl(\openone\otimes\ket{1}\bra{1}\bigr)\cC_{X}\bigr\}$
straightforwardly.
Such rotations for the two adaptive measurements in $X_{\Psi}$ are given by
\begin{equation}
\begin{aligned}
  \bigl(H\otimes R_{+}\bigr) \bigl(\ket{+}\bra{+}\otimes\ket{\varphi_{+}^{\perp}}\bra{\varphi_{+}^{\perp}}\bigr) \bigl(H\otimes R_{+}\bigr)^\dagger&=\ket{0}\bra{0}\otimes\ket{0}\bra{0}\,,\\
  \bigl(H\otimes R_{-}\bigr) \bigl(\ket{-}\bra{-}\otimes\ket{\varphi_{-}^{\perp}}\bra{\varphi_{-}^{\perp}}\bigr) \bigl(H\otimes R_{-}\bigr)^\dagger&=\ket{1}\bra{1}\otimes\ket{0}\bra{0}\,,
\end{aligned}
\end{equation}
where $H$ is the Hadamard gate and $R_{\pm}$ are rotations that turn the state
$\ket{\varphi_{\pm}^{\perp}}$ into \ket{0}.
Specifically, we have
\begin{equation}
  R_{+}=\left[\begin{matrix}\sin\theta&\cos\theta\\-\cos\theta&\sin\theta\end{matrix}\right],
  ~R_{-}=\left[\begin{matrix}\sin\theta&-\cos\theta\\\cos\theta&\sin\theta\end{matrix}\right].
\end{equation}
Then, the adaptive QND measurements for $X_{\Psi}$ can be constructed as
\begin{eqnarray}\label{eq:adap_detail}
  \cM_{X_{\Psi}}
  &=&\Bigl[H^{(1)}\otimes R_{+}^{\dagger(2)}\otimes\openone^{(a,a')}\Bigr]
  \Bigl[\openone^{(1,2)}\otimes\openone^{(a)}\otimes\bigl(\ket{0}\bra{0}\bigr)^{(a')}\Bigr]
  \cC_{X2a'}
  \Bigl[\openone^{(1)}\otimes R_{+}^{(2)}\otimes\openone^{(a,a')}\Bigr]
  \nonumber\\
  &&
  \Bigl[\openone^{(1,2)}\otimes\bigr(\ket{0}\bra{0}\bigl)^{(a)}\otimes\openone^{(a')}\Bigr]
  \cC_{X1a}
  \Bigl[H^{(1)}\otimes\openone^{(2)}\otimes\openone^{(a,a')}\Bigr]\nonumber\\
  &+&
  \Bigl[H^{(1)}\otimes R_{-}^{\dagger(2)}\otimes\openone^{(a,a')}\Bigr]
  \Bigl[\openone^{(1,2)}\otimes\openone^{(a)}\otimes\bigl(\ket{0}\bra{0}\bigr)^{(a')}\Bigr]
  \cC_{X2a'}
  \Bigl[\openone^{(1)}\otimes R_{-}^{(2)}\otimes\openone^{(a,a')}\Bigr]
  \nonumber\\
  &&
  \Bigl[\openone^{(1,2)}\otimes\bigr(\ket{1}\bra{1}\bigl)^{(a)}\otimes\openone^{(a')}\Bigr]
  \cC_{X1a}
  \Bigl[H^{(1)}\otimes\openone^{(2)}\otimes\openone^{(a,a')}\Bigr]
  \,,
\end{eqnarray}
where the superscripts $(1,2)$ and $(a,a')$ denote the two system qubits and two ancilla qubits respectively.
Note that $\cM_{X_{\Psi}}$ contains two terms, because $X_{\Psi}$ is an adaptive measurement
with two branches \cite{Liu.etal2019b}.

Specifically, the whole measurement process consists of two steps.
The first step is to realize the QND version of the two-outcome projective measurement
$\{\ket{+}\bra{+},\ket{-}\bra{-}\}$ on the first particle.
The second is to realize the QND version of the adaptive projective measurement
$\ket{\varphi_{+}^{\perp}}\bra{\varphi_{+}^{\perp}}$ or $\ket{\varphi_{-}^{\perp}}\bra{\varphi_{-}^{\perp}}$
on the second particle according to the outcome of the first step.
To be specific, we rotate the first particle with a Hadamard gate $H$ and send it to a {\sc cnot} gate
$\cC_{X}$ together with an ancilla qubit $a$ in state \ket{0}.
Then we measure the ancilla qubit using a Pauli-$Z$ measurement.
If the outcome is \ket{0} (or \ket{1}), we rotate the second particle by $R_{+}$ (or $R_{-}$)
and send it to another {\sc cnot} gate together with a new ancilla qubit $a'$ in state \ket{0}.
After that, the ancilla is measured by a Pauli-$Z$ measurement, and we declare the test passes
if the outcome \ket{0} is obtained.
Finally, the state is rotated back to its original form.
Note that all the rotations and measurements after the coupling operations are commutative,
so that we can conveniently
choose to do the measurement in the very end.
This results in the alternative writing of Eq.~\eqref{eq:adap_detail}, i.e.,
\begin{eqnarray}
  \cM_{X_{\Psi}}
  &=&
  (\openone\otimes\openone\otimes\ket{00}\bra{00})
  (H\otimes R_{+}^{\dagger}\otimes\openone\otimes\openone)
  \cC_{X2a'}
  (\openone\otimes R_{+}\otimes\openone\otimes\openone)
  \cC_{X1a}
  (H\otimes\openone\otimes\openone\otimes\openone)\nonumber\\
  &+&
  (\openone\otimes\openone\otimes\ket{10}\bra{10})
  (H\otimes R_{-}^{\dagger}\otimes\openone\otimes\openone)
  \cC_{X2a'}
  (\openone\otimes R_{-}\otimes\openone\otimes\openone)\cC_{X1a}(H\otimes\openone\otimes\openone\otimes\openone)\,,
\end{eqnarray}
in which we omit the superscripts by showing the explicit operations performed on all parties.

\section{Appendix F: Proof of Proposition~\ref{prop:2bodyCouple}}\label{app:2bodyCouple}
\begin{proof}
  First, we consider the generalized Toffoli gate $\cC_{X}^{n}$ such that
  the ancilla qubit is controlled by $n$ system qubits, i.e.,
  \begin{eqnarray}
    \cC_{X}^{n}&=&\bigl(\openone^{\otimes n}-\ket{11\cdots1}\bra{11\cdots1}\bigr)\otimes\openone+\ket{11\cdots1}\bra{11\cdots1}\otimes X\,.
  \end{eqnarray}
  Note that the {\sc cnot} gate $\cC_{X}$ represents a special case of $\cC_{X}^{n}$ when $n=1$, namely
  \begin{eqnarray}
    \cC_{X}&=&\ket{0}\bra{0}\otimes\openone+\ket{1}\bra{1}\otimes X\,.
  \end{eqnarray}
  Following Eq.~\eqref{eq:equiv}, for an arbitrary input state $\sigma$ together with an ancilla qubit $\ket{0}$,
  the QND measurement with {\sc cnot} gates has the relation
  \begin{equation}
    \bigl[\bigl(\openone\otimes\ket{0}\bra{0}\bigr)\cC_{X}\bigr]\bigl(\sigma\otimes\ket{0}\bra{0}\bigr)
    =\bigl(\ket{0}\bra{0}\otimes\ket{0}\bra{0}+\ket{1}\bra{1}\otimes\ket{0}\bra{1}\bigr)\bigl(\sigma\otimes\ket{0}\bra{0}\bigr)
    =\bigl(\ket{0}\bra{0}\otimes\openone\bigr)\bigl(\sigma\otimes\ket{0}\bra{0}\bigr)\,,
  \end{equation}
  which tells us that $\bigl(\openone\otimes\ket{0}\bra{0}\bigr)\cC_{X} \widehat{=} \ket{0}\bra{0}\otimes\openone$.
  Then, the sequential measurement constructed using $n$ {\sc cnot} gates gives
  \begin{equation}\label{eq:app_cx}
    \bigl[\bigl(\openone\otimes\ket{0}\bra{0}\bigr)\cC_{X1a_1}\bigr]
    \bigl[\bigl(\openone\otimes\ket{0}\bra{0}\bigr)\cC_{X2a_2}\bigr]
    \cdots
    \bigl[\bigl(\openone\otimes\ket{0}\bra{0}\bigr)\cC_{Xna_n}\bigr]
    \widehat{=}
    \ket{00\cdots0}\bra{00\cdots0}\otimes\openone^{\otimes n}\,,
  \end{equation}
  where we adopt the convention that the first part $\ket{00\cdots0}\bra{00\cdots0}$ operates on the $n$ system qubits
  and the second part $\openone^{\otimes n}$ operates on the $n$ ancilla.
  With the rotation of a Pauli-$X$ measurement, the generalized Toffoli gate becomes
  \begin{equation}
    \bigl(X^{\otimes n} \otimes \openone\bigr)\cC_{X}^{n}\bigl(X^{\otimes n} \otimes \openone\bigr)= \bigl(\openone^{\otimes n}-\ket{00\cdots0}\bra{00\cdots0}\bigr)\otimes\openone+\ket{00\cdots0}\bra{00\cdots0}\otimes X\,,
  \end{equation}
  so that the QND measurement with rotated $\cC_{X}^{n}$ has the equivalence
  \begin{equation}\label{eq:app_cnx}
    \bigl(\openone^{\otimes n}\otimes\ket{0}\bra{0}\bigr)\bigl(X^{\otimes n} \otimes \openone\bigr)\cC_{X}^{n}\bigl(X^{\otimes n} \otimes \openone\bigr)
    \widehat{=}
    \bigl(\openone^{\otimes n}-\ket{00\cdots0}\bra{00\cdots0}\bigr)\otimes\openone\,.
  \end{equation}
  Following Eqs.~\eqref{eq:app_cx} and \eqref{eq:app_cnx}, one can quickly find
  \begin{equation}\label{eq:Toffoli}
  \begin{aligned}
    &\bigl(\openone^{\otimes n}\otimes\ket{0}\bra{0}\bigr)\bigl(X^{\otimes n} \otimes \openone\bigr)\cC_{X}^{n}\bigl(X^{\otimes n} \otimes \openone\bigr)\otimes\openone^{\otimes(n-1)}\\
    &\widehat{=}
    \openone^{\otimes 2n}-
    \bigl[\bigl(\openone\otimes\ket{0}\bra{0}\bigr)\cC_{X1a_1}\bigr]
    \bigl[\bigl(\openone\otimes\ket{0}\bra{0}\bigr)\cC_{X2a_2}\bigr]
    \cdots
    \bigl[\bigl(\openone\otimes\ket{0}\bra{0}\bigr)\cC_{Xna_n}\bigr]\,,
  \end{aligned}
  \end{equation}
  where $\openone^{\otimes(n-1)}$ indicates the addition of $(n-1)$ ancilla qubits needed to implement the gate $\cC^n_X$ with $n$ $\cC_X$ gates.
  Eq.~\eqref{eq:Toffoli} means that $(n+1)$-body Toffoli gate can always be replaced by $n$ two-body {\sc cnot} gates.

  Now we can generalize the proof by considering the $(n+1)$-body coupling in the form of
  \begin{equation}
    R\cC_{X}^{n}:=
    \bigl(R^{(1)\dag} \otimes R^{(2)\dag} \otimes \cdots \otimes R^{(n)\dag} \otimes \openone\bigr)\cC_{X}^{n}\bigl(R^{(1)} \otimes R^{(2)} \otimes \cdots \otimes R^{(n)} \otimes \openone\bigr)
    =\bigl(nR^\dag \otimes \openone\bigr)\cC_{X}^{n}\bigl(nR \otimes \openone\bigr)\,,
  \end{equation}
  where $R^{(i)}$s are arbitrary local unitary operations.
  One notes that as the initial state of the target qubit, which is the ancilla, in the QND measurements is fixed to be \ket{0}, the local operator $R^{(a)}=\openone$.
  Similarly, for a two-body coupling controlled by system qubit $i$ and targeted on ancilla qubit $a$, we have
  \begin{eqnarray}
    R\cC_{Xia}:=\bigl(R^{(i)\dagger} \otimes \openone\bigr)\cC_{Xia}\bigl(R^{(i)} \otimes \openone\bigr)&\widehat{=}&\ket{\phi_i}\bra{\phi_i}\otimes\openone\,,
  \end{eqnarray}
  where $\ket{\phi_i}=R^{(i)\dagger}\ket{0}$.
  So we have the QND measurements
  \begin{eqnarray}
    nR\cM_X&:=&
    \bigl[\bigl(\openone\otimes\ket{0}\bra{0}\bigr)R\cC_{X1a_1}\bigr]
    \bigl[\bigl(\openone\otimes\ket{0}\bra{0}\bigr)R\cC_{X2a_2}\bigr]
    \cdots
    \bigl[\bigl(\openone\otimes\ket{0}\bra{0}\bigr)R\cC_{Xna_n}\bigr]
    \widehat{=}\ket{\phi}\bra{\phi}\otimes\openone^{\otimes n}\,,
  \end{eqnarray}
  with $\ket{\phi}=\bigotimes_i\ket{\phi_i}$.
  Also, we can define the QND measurements
  \begin{eqnarray}
    R\cM_{X}^{n}&:=&
    \bigl(nR^\dag \otimes \openone\bigr)\bigl(X \otimes X \otimes \cdots \otimes X \otimes \openone\bigr)\cC_{X}^{n}\bigl(X \otimes X \otimes \cdots \otimes X \otimes \openone\bigr)\bigl(nR \otimes \openone\bigr)\nonumber\\
    &=&\bigl(nR^\dag \otimes \openone\bigr)\bigl(X \otimes X \otimes \cdots \otimes X \otimes \openone\bigr)\bigl(nR \otimes \openone\bigr)R\cC_{X}^{n}\bigl(nR^\dag \otimes \openone\bigr)\bigl(X \otimes X \otimes \cdots \otimes X \otimes \openone\bigr)\bigl(nR \otimes \openone\bigr)\nonumber\\
    &\widehat{=}&
    \bigl(\openone^{\otimes n}-\ket{\phi}\bra{\phi}\bigr)\otimes\openone\,.
  \end{eqnarray}
  Then one can quickly obtain the more general relation
  \begin{equation}
    R\cM_{X}^{n}\otimes\openone^{\otimes(n-1)} \widehat{=} \openone^{\otimes 2n}-nR\cM_X\,.
  \end{equation}
  Note that the effective commutation between the projective measurement $\openone\otimes\ket{0}\bra{0}$
  and the local operations $R^{i}$ and Pauli-$X$, as well as the commutation between
  $\cC_{Xia}$ and $R^{(j)}$($i\neq j$) can help us move the operators around and simplify the
  experimental realizations.

  For example, the QND measurement using {\sc cnot} gates for verifying arbitrary two-qubit pure states
  with the order of Eq.~\eqref{eq:NDQV_2qb_2} in the main text is
  \begin{equation}
    \cM_{i}^{b}
     =\openone^{\otimes4}-\bigl(\openone^{\otimes 2}\otimes\ket{00}\bra{00}\bigr)\bigl(R_i^{\dagger}\otimes\openone^{\otimes 2}\bigr)\cC_{X1a}\cC_{X2a'}\bigl(R_i\otimes \openone^{\otimes 2}\bigr)\,.
  \end{equation}
  It can easily be realized with four steps using two ancilla qubits initialized as \ket{00}:
  two local rotations, two couplings, two local rotations, and two ancilla measurements.
  However, it can also be rewritten as
  \begin{eqnarray}
    \cM_{i}^{b\prime}
     &=&\openone^{\otimes 4}-
     \Bigl[\bigl(\openone\otimes\openone\otimes\ket{0}\bra{0}\otimes\openone\bigr)
     \bigl(R_{i,1}^{\dagger}\otimes\openone\otimes\openone\otimes\openone\bigr)
     \cC_{X1a}
     \bigl(R_{i,1}\otimes\openone\otimes\openone\otimes\openone\bigr)\Bigl]
     \nonumber\\
     &&\qquad
     \Bigl[\bigl(\openone\otimes\openone\otimes\openone\otimes\ket{0}\bra{0}\bigr)
     \bigl(\openone\otimes R_{i,2}^{\dagger}\otimes\openone\otimes\openone\bigr)
     \cC_{X2a'}
     \bigl(\openone\otimes R_{i,2}\otimes \openone\otimes\openone\bigr)\Bigl]\,,
  \end{eqnarray}
  where $R_i=R_{i,1}\otimes R_{i,2}$.
  This order is in favor of modular designs that can be constructed by the same blocks on the two qubits $(1,a)$ or $(2,a')$
  with four steps: one local rotation, one coupling between the two qubits, one local rotation,
  and one ancilla measurement.
  Such blocks can be conveniently extended for the verification of other multipartite entangled states.
\end{proof}

\section{Appendix G: Verification of GHZ states}
Take the three-qubit GHZ state $\ket{\mbox{GHZ}_3}=(\ket{000}+\ket{111})/\sqrt{2}$,
where one set of the stabilizer generators are $\{XXX,Z\openone Z,ZZ\openone\}$, as an example.
We can construct the sequential NDQV protocol as $\cM=\cM_1 \cM_2 \cM_3$ (or arbitrary permutations
of the three measurement settings) with
\begin{equation}
\begin{aligned}
  \cM_1&=\bigl(\openone\otimes\ket{0}\bra{0}\bigr)\bigl[(H\otimes H\otimes H\otimes \openone)\cC_{X1a}\cC_{X2a}\cC_{X3a}(H\otimes H\otimes H\otimes \openone)\bigr]\,,\\
  \cM_2&=\bigl(\openone\otimes\ket{0}\bra{0}\bigr)\bigl[\cC_{X1a}\cC_{X3a}\bigr]\,,\\
  \cM_3&=\bigl(\openone\otimes\ket{0}\bra{0}\bigr)\bigl[\cC_{X1a}\cC_{X2a}\bigr]\,.
\end{aligned}
\end{equation}
Then, direct calculations can prove that $\cM$ is equivalent to the optimal global strategy.

\end{document}